\documentclass[reqno,11pt]{amsart}

\newtheorem{definition}{Definition}
\usepackage{graphicx}
\usepackage{xcolor}
\usepackage{tikz}
\usepackage{tikz-cd}
\usepackage{caption}

\usepackage{amssymb}

\usepackage{amsmath}

\newcommand{\be}{\begin{equation}}
\newcommand{\ee}{\end{equation}}

\DeclareMathSymbol{\Lambda}{\mathord}{operators}{"03}

\usepackage{enumerate}
\usepackage{enumitem}

\usepackage{appendix}

\newtheorem{thm}{Theorem}[section]
\newtheorem{prop}[thm]{Proposition}

\newtheorem{remark}[thm]{\it Remark}

\usepackage{tikz-3dplot}
\usepackage{xifthen}
\usepackage{caption}

\tdplotsetmaincoords{60}{125}
\tdplotsetrotatedcoords{0}{20}{0} 

\usepackage[
labelfont=sf,
hypcap=false,
format=hang
]{caption}

\begin{document}

\title{On non-abelian quadrirational Yang-Baxter maps}

\author[P. Kassotakis]{Pavlos Kassotakis}
\address{Pavlos Kassotakis, -----Nicosia, Cyprus}
 \email{pavlos1978@gmail.com}

\author[T. Kouloukas]{Theodoros Kouloukas}
\address{Theodoros Kouloukas, School of Mathematics and Physics, University of Lincoln, Lincoln LN6 7TS, U.K.}
\email{TKouloukas@lincoln.ac.uk}

\date{\today}

\begin{abstract}
We introduce four lists of families of non-abelian quadrirational Yang-Baxter maps.
\end{abstract}

\dedicatory{Dedicated to the Memory of   \\Aristophanes Dimakis, 1953 -- 2021}

\maketitle

\vspace{2pc}




\section{Introduction}

In the recent years there is a growing interest in deriving and extending discrete
integrable systems to the non-abelian domain.  At the same time there is an intrinsic connection of discrete
integrable systems with Yang-Baxter maps \cite{ABS:YB,pap2-2006,Tasos:ell,Kassotakis_2012,Atkinson:2012i,Kouloukas:2012,Kouloukas:2015,Atkinson:2018,BAZHANOV2018509,Kass2,Kels:2019II}. Although very important examples of non-commutative Yang-Baxter maps exist in the literature \cite{Kajiwara:2002,Goncharenko:2001,Goncharenko:2004,ABS:YB,veselov-2003,Doliwa_2014,Dimakis:2019,Dimakis_Korepanov:2021}, the non-abelian counterparts of the Harrison map \cite{pap2-2006}, a.k.a. the nonlinear superposition formula for the B\"acklund transformation of the Ernst equation \cite{Harrison:1978},
referred to as $H_I$ in \cite{Papageorgiou:2010}, and of the $F_I$ \cite{ABS:YB} quadrirational Yang-Baxter maps are not known. The Harrison map $H_I$ as well as the $F_I$ map are the canonical forms of the generic maps ({\em top} members) of two non-equivalent lists of families of Yang-Baxter maps with five members each, the $H$-list and the $F-$list respectively. The maps $H_I$ and $F_I$ are considered top members of the corresponding lists since the remaining members can arise through degeneracies.

In this article we provide explicitly the non-abelian avatars 
of the $H_I$ and the $F_I$ Yang-Baxter maps, which participate as the  generic maps of   what we will call the $\mathcal{H}-$list and the $\mathcal{F}-$list respectively. In addition we provide the generic maps of 
two additional non-equivalent lists that we refer to as the   $\mathcal{K}$-list and the $\mathcal{\Lambda}-$list. The generic members of  $\mathcal{K}$-list and the $\mathcal{\Lambda}-$list in the abelian setting are both equivalent to the $H_I$ Yang-Baxter map, which is not the case in the non-abelian setting as we  show.  We start this article with a short introduction followed by Section \ref{sec2} where the basic definitions used throughout this paper are introduced. In addition, we provide the non-abelian version of the so-called Adler map as a motivating example. In Section \ref{sec3}, via a Lax pair formulation, we derive the non-abelian versions of the the top members of the $\mathcal{K},\mathcal{\Lambda},\mathcal{H}$ and $\mathcal{F}$ lists of families of quadrirational Yang-Baxter maps. In full extend these lists are presented in Appendix \ref{app1}. Finally, we conclude this article in Section \ref{sec4} where we present some ideas for further research.

\section{Definitions and a non-abelian extension of the Adler map} \label{sec2}
Let $\mathbb{X}$ be any set. We proceed with the following definitions.

\begin{definition}
The  maps $R: \mathbb{X} \times \mathbb{X} \rightarrow \mathbb{X} \times \mathbb{X}$ and $\widehat R: \mathbb{X} \times \mathbb{X} \rightarrow \mathbb{X} \times \mathbb{X}$  will be called $YB$ equivalent if it exists a bijection $\kappa: \mathbb{X} \rightarrow \mathbb{X}$ such that $(\kappa\times \kappa) \, R=\widehat R \, (\kappa\times \kappa).$
\end{definition}

\begin{definition}
A bijection $\phi: \mathbb{X}\rightarrow \mathbb{X}$ will be called {\em symmetry} of the    map $R: \mathbb{X} \times \mathbb{X}\rightarrow \mathbb{X} \times \mathbb{X},$ if  $(\phi\times \phi) \, R= R \, (\phi\times \phi).$
\end{definition}

\begin{definition}[Yang-Baxter map]
A map $R: \mathbb{X} \times \mathbb{X}\ni(u,v)\mapsto (x,y)=(x(u,v),y(u,v))\in \mathbb{X}\times \mathbb{X},$ will be called {\em Yang-Baxter map} if it satisfies the {\em Yang-Baxter relation}
\begin{align*}
R_{12}\circ R_{13}\circ R_{23}= R_{23}\circ R_{13}\circ R_{12},
\end{align*}
where $R_{i,j}$ $i,j\in\{1,2,3\},$ denotes the action of the map $R$ on the $i-$th and the $j-$th factor of $\mathbb{X}\times \mathbb{X}\times \mathbb{X},$ i.e.
$R_{12}:(u,v,w)\mapsto (x(u,v),y(u,v),w),$ $R_{13}:(u,v,w)\mapsto (x(u,w),v,z(u,w)),$ and $R_{23}:(u,v,w)\mapsto (u,y(v,w),z(v,w)).$
\end{definition}
Alternatively we can use the  definition of {\em 3D-compatible maps} \cite{ABS:YB}. Let $F: \mathbb{X} \times \mathbb{X}\mapsto \mathbb{X} \times \mathbb{X},$ be a map and  $F_{ij}$  $i<j\in\{1,2,3\},$ be the maps that act as $F$ on the $i-$th and $j-$th factor of $\mathbb{X} \times \mathbb{X}\times \mathbb{X}.$   For this definition it is convenient to denote these maps in components   as follows
\begin{align*}
F_{ij}: (u^i,u^j,u^k)\mapsto (u^i_j,u^j_i,u^k)=(u^i_j(u^i,u^j),u^j_i(u^i,u^j),u^k),&& i\neq j\neq k\neq i\in\{1,2,3\}.
\end{align*}
\begin{definition}[3D-compatible map \cite{ABS:YB}]
A map $F: \mathbb{X} \times \mathbb{X}\rightarrow \mathbb{X} \times \mathbb{X}$ will be called {\em 3D-compatible map} if it holds $u^{i}_{jk}=u^i_{kj}$ i.e.
\begin{align*}
u^i_j\left(u^i_k(u^i,u^k),u^j_k(u^j,u^k)\right)=u^i_k\left(u^i_j(u^i,u^j),u^k_j(u^k,u^j)\right),&& i\neq j\neq k\neq i\in\{1,2,3\}.
\end{align*}
\end{definition}

\begin{remark}
$YB$ equivalency respects the Yang-Baxter property as well as 3D-compatibility.
\end{remark}

The following proposition  was considered in \cite{Papageorgiou:2010}.

\begin{prop} \label{prop:sec2:1}
Let $\phi: \mathbb{X}\rightarrow \mathbb{X}$ a symmetry of the Yang-Baxter map $R:\mathbb{X}\times \mathbb{X} \rightarrow \mathbb{X}\times \mathbb{X}.$ Then the map
\begin{align*}
{\widehat R}=(\phi^{-1}\times id)\,R\,(id\times\phi),
\end{align*}
is also a Yang-Baxter map.
\end{prop}
Clearly the Yang-Baxter maps  $R$ and ${\widehat R}$ are not $YB$ equivalent. Hence, given a Yang-Baxter map and a symmetry of this map, one can introduce another  Yang-Baxter map not equivalent with the original. Note that the same holds true for $3D-$compatible maps.

\begin{definition}[\cite{et-2003,ABS:YB}]
A map $R: \mathbb{X} \times \mathbb{X}\ni(u,v)\mapsto (x,y) \in \mathbb{X} \times \mathbb{X}$   will be called quadrirational, if both the map $R$ and the so called {\em companion map} $cR: \mathbb{X} \times \mathbb{X}\ni(x,v)\mapsto (u,y) \in \mathbb{X} \times \mathbb{X},$ are birational maps.
\end{definition}
The notion of {\em quadrirational maps} first appeared in \cite{et-2003} under the name {\em non-degenerate rational maps}. Non-degenerate rational maps where renamed {\em quadrirational maps} in \cite{ABS:YB}.

\begin{remark}
The companion map of a quadrirational Yang-Baxter map is a 3D-compatible map.
The converse also holds i.e. the companion map of a quadrirational 3D-compatible map is a Yang-Baxter map.
\end{remark}

\begin{definition}[\cite{Veselov:2003b,Nijhoff:2002}]
The matrix $L(x;\lambda)$ \begin{enumerate}
\item is called a Lax matrix of the Yang-Baxter map $R: (u,v)\mapsto (x,y),$  if the relation
$
R(u,v)=(x,y)
$
implies that
$
L(u;\lambda)L(v;\lambda)=L(y;\lambda)L(x;\lambda)
$
 for all $\lambda$; 
\item is called a Lax matrix of the  companion map $cR: (x,v)\mapsto (u,y),$  if the relation
$
cR(x,v)=(u,y)
$
implies that
$
L(u;\lambda)L(v;\lambda)=L(y;\lambda)L(x;\lambda)
$
 for all $\lambda$. $L(x;\lambda)$ is called a strong Lax matrix of $cR$ if the converse also holds.
\end{enumerate}
\end{definition}

For the rest of the article we consider the set $\mathbb{X}$ to be  $\mathbb{D}\times \mathbb{D},$ where $\mathbb{D}$ is a non-commutative division ring, i.e.  an associative algebra
with a multiplicative identity
element denoted by $1$ and every non-zero element $x$ of $\mathbb{D}$ has a unique multiplicative inverse denoted
by $x^{-1}$ s.t. $x x^{-1} = x^{-1} x = 1$.

\subsection{Non-abelian extension of the Adler map}

A prototypical example of a Yang-Baxter map on $\mathbb{CP}^1\times\mathbb{CP}^1$ is the so-called Adler map (or $H_V$) 
 that was introduced in \cite{Adler:1993}. The Adler map reads:
\begin{align*}
H_V:(u,v)\mapsto (x,y)=\left(v-\frac{p-q}{u+v},u+\frac{p-q}{u+v}\right).
\end{align*}
For the rest of this Section and in the Propositions that follow, we extend the Adler map and its companion on $\mathbb{X}\times\mathbb{X},$ where $\mathbb{X}:=\mathbb{D}\times\mathbb{D}.$

\begin{prop} \label{prop01}
 The map  $
c\mathcal{H}_V: (x^1,x^2,v^1,v^2)\mapsto (u^1,u^2,y^1,y^2),
$
where
\begin{align*}
u^1=\left(v^1(v^1-x^1)+v^2v^1-x^2x^1\right) \left(x^1-v^1\right)^{-1},& &
u^2=(v^1-x^1)x^2x^1\left((v^1+x^2)x^1-(v^1+v^2)v^1\right)^{-1},\\
y^1=\left(x^1(x^1-x^1)+x^2x^1-v^2v^1\right) \left(v^1-x^1\right)^{-1},& &
y^2=(x^1-v^1)v^2v^1\left((x^1+v^2)v^1-(x^1+x^2)x^1\right)^{-1},
\end{align*}
\begin{enumerate}
\item has as symmetry the bijection $\phi: (z^1,z^2)\mapsto (-z^1,-z^2)$;
\item
has as strong Lax matrix
the matrix
\begin{align} \label{Lax_Hv}
L(x^1,x^2;\lambda)=
\begin{pmatrix}
x^1&(x^1+x^2)x^1-\lambda\\
1&x^1
\end{pmatrix},
\end{align}
where we assume that the spectral parameter $\lambda$  belongs to the center $C(\mathbb{D})$ of the algebra $\mathbb{D},$ i.e. it commutes with every element of $\mathbb{D}$;
\item  it is a  $3D-$compatible map.
\end{enumerate}
\end{prop}
\begin{proof}
Is is easy to verify that $(\phi\times\phi)\,cH_V=cH_V\,(\phi\times\phi)$ and that proves that the bijection $\phi$ is a symmetry of $c\mathcal{H}_V$.
As a consequence,  the map $(u^i,v^i,x^i,y^i)\mapsto (-u^i,-v^i,-x^i,-y^i),$ $\forall i\in\{1,2\},$  leaves invariant the compatibility conditions (\ref{eq1})-(\ref{eq4}).

Let us now prove item $(2).$
From the Lax equation $L(u^1,u^2;\lambda)L(v^1,v^2;\lambda)=L(y^1,y^2;\lambda)L(x^1,x^2;\lambda),$ we obtain the following compatibility conditions
\begin{flalign} \label{eq1}
&u^1+v^1=y^1+x^1,\\ \label{eq2}
&u^1v^1+(v^1+v^2)v^1=y^1x^1+(x^1+x^2)x^1,\\ \label{eq3}
&(u^1+u^2)u^1+u^1v^1=(y^1+y^2)y^1+y^1x^1,\\ \label{eq4}
&(u^1+u^2)u^1v^1+u^1(v^1+v^2)v^1=(y^1+y^2)y^1x^1+y^1(x^1+x^2)x^1.
\end{flalign}
First note that the compatibility conditions (\ref{eq1})-(\ref{eq4}) are symmetric under the interchange
\begin{flalign}\label{sym01}
&(u^1,u^2,v^1,v^2)\leftrightarrow (y^1,y^2,x^1,x^2).
\end{flalign}
 Equations (\ref{eq1}), (\ref{eq2}) are linear in  $u^1,y^1,$ and do not include $u^2,y^2,$  so the latter can be easily solved for $u^1,y^1.$ Specifically, by eliminating $y^1$ from (\ref{eq1}), (\ref{eq2}), we obtain
 \begin{flalign}\label{eqq1}
 &u^1=\left(v^1(v^1-x^1)+v^2v^1-x^2x^1\right) \left(x^1-v^1\right)^{-1}.
  \end{flalign}
Applying (\ref{sym01}) to (\ref{eqq1}) we get
 \begin{flalign}\label{eqq2}
 &y^1=\left(x^1(x^1-v^1)+x^2x^1-v^2v^1\right) \left(v^1-x^1\right)^{-1}.
  \end{flalign}
  Substituting (\ref{eqq1}),(\ref{eqq2}) to (\ref{eq3}),(\ref{eq4}) and by solving them we obtain
 \begin{flalign}\label{eQ3}
 &y^2=(x^1-v^1)v^2v^1\left(x^1(v^1-x^1)+v^2v^1-x^2x^1\right)^{-1},
  \end{flalign}
 \begin{flalign}\label{eQ4}
 &u^2=(v^1-x^1)x^2x^1\left(v^1(x^1-v^1)+x^2x^1-v^2v^1\right)^{-1}.
  \end{flalign}
Equations (\ref{eqq1}),(\ref{eqq2}), (\ref{eQ3}) and (\ref{eQ4}), coincide with the defining relations of the map $c\mathcal{H}_V$ of the proposition. Moreover,  (\ref{eqq1}),(\ref{eqq2}), (\ref{eQ3}) and (\ref{eQ4}) is the unique solution of the compatibility conditions (\ref{eq1})-(\ref{eq4}), so the Lax matrix $L(x^1,x^2;\lambda)$ is strong.

 The $3D-$compatibility of $c\mathcal{H}_V$ can be proven by direct computation. Alternatively, by using the fact that $L(x^1,x^2;\lambda)$ is strong, from \cite{Nijhoff:2002,bs:2002N} the $3D-$compatibility follows.
\end{proof}

In the commutative setting where all variables are considered elements of the center of the algebra $\mathbb{D},$  from the defining relations of $c\mathcal{H}_V$  we obtain that
\begin{flalign*}
&u^2u^1=x^2x^1,\quad y^2y^1=v^2v^1,
\end{flalign*}
so the products $x^2x^1$ and $v^2v^1$ are invariants of $cH_V$. Clearly the latter are no longer invariants of the map in the non-commutative setting. Nevertheless, if we assume that the products $x^2x^1$ and $v^2v^1$ belong to the center of the algebra $\mathbb{D},$ i.e. they commute with all elements of $\mathbb{D},$ then these products are invariants of the map. This assumption is referred to as {\em centrality assumption} and it was firstly introduced in \cite{Doliwa_2013,Doliwa_2014} where it played an essential role in obtaining the companion map of the so-called {\em N-periodic reduction of the KP map}.

From further on,  when we refer to the centrality assumption for a map  $F: (z^1,z^2,w^1,w^2) \mapsto (\bar z^{1},\bar z^{2},\bar w^{1},\bar w^{2})$ we refer to
\begin{align}\label{cent_ass}
z^2z^1=p\in C(\mathbb{D}),\quad w^2w^1=q\in C(\mathbb{D}),
\end{align}
where $C(\mathbb{D})$ the center of the algebra $\mathbb{D}.$ Note that as a consequence of (\ref{cent_ass}), we have the commutativity relations $z^2z^1=z^1z^2,$ $w^2w^1=w^1w^2.$ 

\begin{prop} \label{prop2}
Under the centrality assumption, the map  $cH_V$ of Proposition \ref{prop01} is quadrirational with  companion map that reads  $\mathcal{H}_V:(u^1,u^2,v^1,v^2)\mapsto (x^1,x^2,y^1,y^2),$ where
\begin{flalign*}
x^1=(u^1+v^1)^{-1}\left((u^1+v^1)v^1+v^2v^1-u^2u^1\right),& &
x^2=u^2u^1\left((u^1+v^1)v^1+v^2v^1-u^2u^1\right)^{-1}(u^1+v^1),\\
y^1=(u^1+v^1)^{-1}\left((u^1+v^1)u^1+u^2u^1-v^2v^1\right),& &
y^2=v^2v^1\left((u^1+v^1)u^1+u^1u^1-v^2v^1\right)^{-1}(u^1+v^1).
\end{flalign*}
The map $\mathcal{H}_V$ is a Yang-Baxter map.
\end{prop}

\begin{proof}
Note that under the centrality assumption (\ref{cent_ass}), from the map  $c\mathcal{H}_V$ of Proposition \ref{prop01} we obtain
\begin{align} \label{cent_ass2}
x^2x^1=u^2u^1=p\in C(\mathbb{D}),\quad y^2y^1=v^2v^1=q\in C(\mathbb{D}).
\end{align}

The first defining relation of $c\mathcal{H}_V$ of Proposition \ref{prop01} reads
\begin{align*}
u^1=-v^1+(v^2v^1-x^2x^1)(x^1-v^1)^{-1},
\end{align*}
  by using  $x^2x^1=u^2u^1$ from (\ref{cent_ass2}), we can solve for $x^1$ to obtain
\begin{align}\label{peq1}
x^1=(u^1+v^1)^{-1}\left((u^1+v^1)v^1+v^2v^1-u^2u^1\right),
\end{align}
namely the first defining relation of $\mathcal{H}_V$ mapping. Now we can substitute (\ref{peq1}) to the second  defining relation of $c\mathcal{H}_V$ and solve for $x^2$ in terms of $u^i,v^i$, $i=1,2,$ or equivalently from (\ref{eq2}) of the compatibility conditions, by using again $x^2x^1=u^2u^1,$ we obtain the second defining relation of $\mathcal{H}_V$ mapping, namely
\begin{align}\label{peq2}
x^2=u^2u^1\left((u^1+v^1)v^1+v^2v^1-u^2u^1\right)^{-1}(u^1+v^1).
\end{align}
Now we substitute (\ref{peq1}) and (\ref{peq2}) to the third and forth defining relation of $c\mathcal{H}_V$ of Proposition \ref{prop01}, to obtain
\begin{align} \label{peq3}
y^1=(u^1+v^1)^{-1}\left((u^1+v^1)u^1+u^2u^1-v^2v^1\right),\\ \label{peq4}
y^2=v^2v^1\left((u^1+v^1)u^1+u^1u^1-v^2v^1\right)^{-1}(u^1+v^1).
\end{align}
 (\ref{peq1})-(\ref{peq4}) constitute the defining relations of $H_V$ mapping and that completes the first part of the proof.

The proof that $\mathcal{H}_V$ is a Yang-Baxter map follows from the fact that it is the companion of a $3D-$compatible map.

\end{proof}
The $\mathcal{H}_V$ map  serves as the non-abelian form of the  Adler map $(H_V)$. This is apparent since under the change of variables $(u,p,v,q)=(u^1,u^1u^2,v^1,v^1v^2),$  $\mathcal{H}_V$ reads
$$
\mathcal{H}_V:(u,p,v,q)\mapsto (x,p,y,q),
$$
where
\begin{align*}
x=v+(u+v)^{-1}(q-p),\quad
y=u+(u+v)^{-1}(p-q),
\end{align*}
that clearly coincides with the Adler map in the commutative case.

\section{Non-abelian extension of  quadrirational Yang-Baxter maps}  \label{sec3}

\begin{prop} \label{prop1}
Provided the centrality assumptions (\ref{cent_ass}), the map  $$
c\mathcal{K}_{a,b,c}: (x^1,x^2,v^1,v^2)\mapsto (u^1,u^2,y^1,y^2),
$$
where
\begin{align*}
u^1=(b-cv^1)(x^2-v^2)x^1\left(x^1-v^1\right)^{-1}\left(a-cv^2\right)^{-1},\\
u^2=(a-cv^2)(x^1-v^1)x^2\left(x^2-v^2\right)^{-1}\left(b-cv^1\right)^{-1},\\
y^1=(b-cx^1)(x^2-v^2)v^1\left(x^1-v^1\right)^{-1}\left(a-cx^2\right)^{-1},\\
y^2=(a-cx^2)(x^1-v^1)v^2\left(x^2-v^2\right)^{-1}\left(b-cx^1\right)^{-1},
\end{align*}
with $a,b,c\in C(\mathbb{D})$ and neither $a,c$ nor $b,c$  simultaneously zero,
\begin{enumerate}
\item has as symmetries the bijections
\begin{align}
&\psi: (z^1,z^2)\mapsto (\frac{b}{a}z^2,\frac{a}{b}z^1),\\
&\phi: (z^1,z^2)\mapsto  \left(\frac{b}{a}(a-cz^2)z^1(cz^1-b)^{-1},\frac{a}{b}(b-cz^1)z^2(cz^2-a)^{-1}\right); 
\end{align}
\item
has as strong Lax matrix
the matrix
\begin{align} \label{Lax_K}
L(x^1,x^2;\lambda)=
\begin{pmatrix}
ax^1-cx^2x^1&\lambda(b-cx^1)\\
a-cx^2&bx^2-cx^1x^2\end{pmatrix},
\end{align}
where the spectral parameter $\lambda\in C(\mathbb{D})$;
\item it is a  $3D-$compatible map;
\item it is quadrirational and its companion map reads
\begin{align*}
\mathcal{K}_{a,b,c}: (u^1,u^2,v^1,v^2)\mapsto (x^1,x^2,y^1,y^2),
\end{align*}
where
\begin{align*}
&x^1=\left(au^1+bv^2-c(v^1v^2+u^1v^2) \right)^{-1}u^1 \left(av^1+bu^2-c(v^2v^1+u^2v^1) \right),\\
&x^2=\left(bu^2+av^1-c(v^2v^1+u^2v^1) \right)^{-1}u^2 \left(bv^2+au^1-c(v^1v^2+u^1v^2) \right),\\
&y^1=\left(au^1+bv^2-c(u^1u^2+u^1v^2) \right)v^1 \left(bu^2+av^1-c(u^2u^1+u^2v^1) \right)^{-1},\\
&y^2=\left(bu^2+av^1-c(u^2u^1+u^2v^1) \right)v^2 \left(au^1+bv^2-c(u^1u^2+u^1v^2) \right)^{-1}.
\end{align*}
\item The map $\mathcal{K}_{a,b,c}$ is a Yang-Baxter map.
\end{enumerate}
\end{prop}
\begin{proof}

First note that as a consequence of the centrality assumption (\ref{cent_ass}), from the map  $c\mathcal{K}_{a,b,c}$  we obtain
\begin{align} \label{cent_ass3}
x^2x^1=u^2u^1=p\in C(\mathbb{D}),\quad y^2y^1=v^2v^1=q\in C(\mathbb{D}).
\end{align}

One can verify  that $(\phi\times\phi)\,c\mathcal{K}_{a,b,c}=c\mathcal{K}_{a,b,c}\,(\phi\times\phi),$ as well as $(\psi\times\psi)\,c\mathcal{K}_{a,b,c}=c\mathcal{K}_{a,b,c}\,(\psi\times\psi)$  and that proves that the bijections $\phi$ and $\psi$ are a symmetries of $c\mathcal{K}_{a,b,c}$. It can also easily shown that $\psi^2=\phi^2=id,$ so $\psi^{-1}=\psi,$ $\phi^{-1}=\phi,$ provided  (\ref{cent_ass3}) holds.

From the Lax equation $L(u^1,u^2;\lambda)L(v^1,v^2;\lambda)=L(y^1,y^2;\lambda)L(x^1,x^2;\lambda),$ we obtain the following compatibility conditions
\begin{align} \label{IIeq1}
&(b-cu^1)(a-cv^2)=(b-cy^1)(a-cx^2),\\ \label{IIeq2}
&(a-cu^2)(b-cv^1)=(a-cy^2)(b-cx^1),\\ \label{IIeq3}
&(b-cu^1)u^2(b-cv^1)v^2=(b-cy^1)y^2(b-cx^1)x^2,\\ \label{eqII4}
&(a-cu^2)u^1(a-cv^2)v^1=(a-cy^2)y^1(a-cx^2)x^1,\\ \label{eqII5}
&(b-cu^1)u^2(a-cv^2)+(a-cu^2)(a-cv^2)v^1=(b-cy^1)y^2(a-cx^2)+(a-cy^2)(a-cx^2)x^1,\\ \label{eqII6}
&(a-cu^2)u^1(b-cv^1)+(b-cu^1)(b-cv^1)v^2=(a-cy^2)y^1(b-cx^1)+(b-cy^1)(b-cx^1)x^2.
\end{align}
The system of equations (\ref{IIeq1})-(\ref{eqII4}) 
determines uniquely $u^i,y^i$ as functions of $v^i,x^i,$ $i=1,2$ i.e. the defining relations of the map $c\mathcal{K}_{a,b,c}$ of this Proposition. Then,  substituting $c\mathcal{K}_{a,b,c}$ into (\ref{eqII5}),(\ref{eqII6}), the latter are satisfied provided  (\ref{cent_ass3}) holds. So $c\mathcal{K}_{a,b,c}$ is uniquely determined by  (\ref{IIeq1})-(\ref{eqII6}) and  that proves that (\ref{Lax_K}) serves as a strong Lax matrix of $c\mathcal{K}_{a,b,c}$.

Using \cite{Veselov:2003b}, the proof that $\mathcal{K}_{a,b,c}$ is a Yang-Baxter map, follows directly  from its Lax representation  and the fact that it can be written as the  projective action
\begin{align*}
x^1=[u^1]L(v^1,v^2,u^2u^1),\;\;x^2=[u^2]L(v^1,v^2,u^1u^2),\;\;y^1=L(u^1,u^2,v^2v^1)[v^1],\;\;y^2=L(u^1,u^2,v^2v^1)[v^2],
\end{align*}
where 
\begin{align*}
\begin{pmatrix}
a&b\\
c&d\end{pmatrix}[x]:=(ax+b)(cx+d)^{-1},\quad [x]\begin{pmatrix}
a&b\\
c&d\end{pmatrix}:=(x c+d)^{-1}(x a+b).
\end{align*}
Then the $3D-$compatibility of $c\mathcal{K}_{a,b,c}$  follows, since it serves as the companion map of a Yang-Baxter map.

\end{proof}

The following remarks are in order.
\begin{itemize}
\item For $a,b,c\geq 0,$ the map $\mathcal{K}_{a,b,-c},$ is a totaly positive Yang-Baxter map.
\item Under the change of variables $(u,p,v,q)=(u^1,u^1u^2,v^1,v^1v^2),$ the symmetries $\phi, \psi,$ obtain the form:
\begin{align} \label{sym1}
&\psi: (u,p)\mapsto (\frac{b}{a}p u^{-1},p),\\ \label{sym2}
&\phi: (u,p)\mapsto  \left(\frac{b}{a}(a u-c p)(cu-b)^{-1},p\right),
\end{align}
also  $\mathcal{K}_{a,b,c}$  reads:
\begin{align} \label{ybK}
\mathcal{K}_{a,b,c}: (u,p,v,q)\mapsto (x,p,y,q),
\end{align}
 where
\begin{align*}
&x=v\left(a u v+bq-cq(u+v)\right)^{-1} \left(a u v+bp-c(qu+pv)\right),\\
&y=\left(auv+bq-c(pv+qu)\right)\left(auv+bp-cp(u+v)\right)^{-1}u.
\end{align*}
In the following Proposition, we use the symmetries $\phi$ and $\psi,$ in order to obtain three additional non $YB$ equivalent families of Yang-Baxter maps.
\end{itemize}

\begin{prop} \label{prop:t}
Let  $\mathcal{K}_{a,b,c}: (u,p,v,q)\mapsto (x,p,y,q),$ be the Yang-Baxter map given in (\ref{ybK}), with symmetries the bijections  $\psi,\phi$ given in (\ref{sym1}),(\ref{sym2}). The maps $\mathcal{\Lambda}_{a,b,c}:=(\psi^{-1}\times id)\,\mathcal{K}_{a,b,c}\; (id \times \psi),$ $\mathcal{H}_{a,b,c}:=(\phi^{-1}\times id)\,\mathcal{K}_{a,b,c}\; (id \times \phi),$ and $\mathcal{F}_{a,b,c}:=(\psi^{-1}\circ\phi^{-1}\times id)\,\mathcal{K}_{a,b,c}\; (id \times \phi\circ \psi),$ 
where
\begin{align*}
\begin{aligned}
&x=pv\left(ab(qu+pv)-cq(bp+auv)\right)^{-1}\left(ab(u+v)-c(bq+auv)\right),\\
&y=q\left(ab(u+v)-c(bp+auv)\right)\left(ab(qu+pv)-cp(bq+auv)\right)^{-1}u,
\end{aligned}&& (\mathcal{\Lambda}_{a,b,c})\\[7pt]
\begin{aligned}
&x=\left((auv-bq)(v-\frac{c}{a}q)^{-1}-(auv-bp)(v-\frac{b}{c})^{-1}\right)^{-1}\\
&\quad \qquad \left(p(auv-bq)(\frac{a}{c}v-q)^{-1}-(auv-bp)(\frac{c}{b}v-1)^{-1}\right),\\
&y=\left(a (ab-c^2q)uv+abc(q-p)v+bq(c^2p-ab)\right)\\
  &\qquad\left(a (ab-c^2p)uv+abc(p-q)u+bp(c^2q-ab)\right)^{-1}u,
\end{aligned}&& (\mathcal{H}_{a,b,c})\\[7pt]
\begin{aligned}
&x=p\left(c p(u-v)(b-cv)^{-1}-a(qu-pv)(cq-av)^{-1}\right)^{-1}\\
  &\quad \qquad  \left(b(u-v)(b-cv)^{-1}-c(qu-pv)(cq-av)^{-1}\right),\\
&y=q\left((ab-c^2q)u+bc(q-p)+(c^2p-ab)v\right)\\
 &\quad \qquad\left(q(ab-c^2p)u+ac(p-q)uv+p(c^2q-ab)v\right)^{-1}u,
\end{aligned}&& (\mathcal{F}_{a,b,c})\\
\end{align*}
are non-abelian quadrirational Yang-Baxter maps.
\end{prop}

\begin{proof}
The proof follows  directly by applying Proposition \ref{prop:sec2:1} to the map  $\mathcal{K}_{a,b,c}$.
\end{proof}

In the abelian setting and for generic $a,b,c$, it holds  that $\mathcal{H}_{a,b,c}$ is related to $\mathcal{K}_{a,b,c}$ through the conjugation $\chi:z\mapsto (1-z)^{-1}$ followed by $(p,q)\mapsto ((1-p)^{-1},(1-q)^{-1})$ \cite{Papageorgiou:2010}, i.e. $\mathcal{H}_{a,b,c}=\chi^{-1}\times\chi^{-1}\,\mathcal{K}_{a,b,c}\,\chi\times \chi.$ Also, $\mathcal{\Lambda}_{a,b,c}$ is related to $\mathcal{K}_{a,b,c}$ through the conjugation $\omega:z\mapsto z^{-1}$ followed by $(p,q)\mapsto (p^{-1},q^{-1}).$ So in the commutative setting essentially we have two non-equivalent families of Yang-Baxter maps, the family $\mathcal{H}_{a,b,c}$ and the family $\mathcal{F}_{a,b,c}$ which coincide with the families $H$ and $F$ in \cite{ABS:YB,Papageorgiou:2010}. In the non-abelian setting,  the families $\mathcal{K}_{a,b,c},$ $\mathcal{\Lambda}_{a,b,c}$ and $\mathcal{H}_{a,b,c}$ are no longer equivalent under conjugation. In the  commutative diagram of Figure \ref{fig1}, we present the four families of quadrirational Yang-Baxter maps $\mathcal{K}_{a,b,c},\mathcal{\Lambda}_{a,b,c},\mathcal{H}_{a,b,c},\mathcal{F}_{a,b,c}$ and their interrelations in the abelian and the non-abelian setting respectively, for generic $a,b$ and $c$.
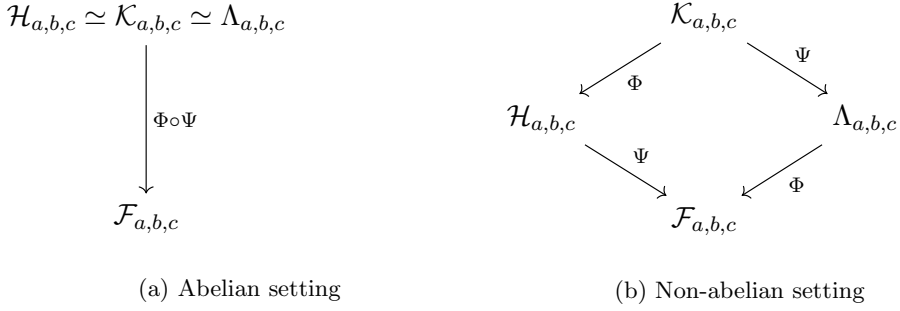
\begin{figure}[h]
\begin{center}
\begin{minipage}[htb]{0.4\textwidth}
\begin{tikzcd}[every arrow/.append style={},row sep=12ex]
\mathcal{H}_{a,b,c}\simeq \mathcal{K}_{a,b,c} \simeq \mathcal{\Lambda}_{a,b,c} \arrow{d}{\Phi\circ \Psi}\\
\mathcal{F}_{a,b,c}
\end{tikzcd}
\captionsetup{font=footnotesize}
\captionof*{figure}{(a) Abelian setting}
\end{minipage}
\begin{minipage}[htb]{0.4\textwidth}
\begin{tikzcd}[every arrow/.append style={}]
   &  \mathcal{K}_{a,b,c}\arrow{ld}{\Phi} \arrow{dr}{\Psi}  & \\
\mathcal{H}_{a,b,c} \arrow{dr}{\Psi} & & \mathcal{\Lambda}_{a,b,c}\arrow{dl}{\Phi}\\
    & \mathcal{F}_{a,b,c}  &
\end{tikzcd}
\captionsetup{font=footnotesize}
\captionof*{figure}{(b) Non-abelian setting}
\end{minipage}\
\caption{The four families of quadrirational Yang-Baxter maps in the  abelian and in the non-abelian setting. The morphisms $\Phi, \Psi$, are  respectively defined by $\Phi: R\rightarrow (\phi^{-1}\times id)R(id\times \phi)$ and $\Psi: R\rightarrow (\psi^{-1}\times id )R(id\times \psi),$ where $\phi,\psi$ the symmetries defined in (\ref{sym1}),(\ref{sym2}).} \label{fig1}
\end{center}
\end{figure}

As a final remark, note that for $a,b,c\geq 0,$ the map $\mathcal{\Lambda}_{a,b,-c},$ is a totaly positive Yang-Baxter~map.

\section{Conclusions}  \label{sec4}

In this article we used Lax formulation to introduce four lists of non-abelian quadrirational Yang-Baxter maps, by providing the canonical forms of the generic maps that correspond to each of these lists, namely the maps $\mathcal{F}_{1,1,1},$ $\mathcal{H}_{1,1,1},$ $\mathcal{K}_{1,1,1}$ and $\mathcal{\Lambda}_{1,1,1}.$ Various degenerations of these canonical maps led to the remaining members of the corresponding $\mathcal{F},$ $\mathcal{H},$ $\mathcal{K}$ and $\mathcal{\Lambda}$ lists of non-abelian Yang-Baxter maps presented in Appendix \ref{app1}.

In the case of {\em entwining Yang-Baxter maps} \cite{Kouloukas:2011}, one can combine the results of this article together with the results of \cite{Kassotakis:2019}, to obtain non-abelian entwining Yang-Baxter maps associated with each member (apart $\mathcal{F}_{IV}$) of the $\mathcal{F},$ $\mathcal{H},$ $\mathcal{K}$ and $\mathcal{\Lambda}$ lists.

Finally note that  the families of Yang-Baxter maps presented in this article serve as the lowest members of hierarchies of families of Yang-Baxter maps. For example the map $\mathcal{K}_{a,b,c}$  is  the companion map (for N=2)  of the following 3D-compatible hierarchy of maps \cite{Kassotakis:2021:p}:
\begin{align*}
c\mathcal{K}^N_{a^1,\ldots,a^N,c}: (x^1,\ldots,x^N,v^1,\ldots,v^N)\mapsto (u^1,\ldots,u^N,y^1,\ldots,y^N),
\end{align*}
where
\begin{align*}
\begin{aligned}
u^i=(a^i-cv^i)(x^{i-1}-v^{i-1})x^i\left(x^i-v^i\right)^{-1}\left(a^{i-1}-cv^{i-1}\right)^{-1},\\
y^i=(a^i-cx^i)(x^{i-1}-v^{i-1})v^i\left(x^i-v^i\right)^{-1}\left(a^{i-1}-cx^{i-1}\right)^{-1},
\end{aligned}&& i=1,\ldots,N,
\end{align*}
with $a^i,c\in C(\mathbb{D})$ and $a^i,c$ $\forall i\in\{1,\ldots,N\}$  not simultaneously zero and the index $i$ is considered modulo $N$. Note that 
\begin{align*}
\psi_{a^1,\ldots,a^N}: (z^1,\ldots,z^i,\ldots,z^N)\mapsto \left(\frac{a^1}{a^N}z^N,\ldots,\frac{a^{i}}{a^i-1}z^{i-1},\ldots,\frac{a^{N}}{a^{N-1}}z^{N-1}\right),
\end{align*}
is a symmetry of $c\mathcal{K}^N_{a^1,\ldots,a^N,c}.$
The hierarchy $c\mathcal{K}^N_{a^1,\ldots,a^N,c}$ with $c=0,$ and $a^i=1, \forall i,$ i.e.  $c\mathcal{K}^N_{1,\ldots,1,0},$ coincides with the so-called $N-$periodic reduction of the KP map and it was firstly considered in \cite{Doliwa_2013,Doliwa_2014} c.f. \cite{Doliwa_Noumi:2020}.  Whereas the non-equivalent hierarchy  $(\psi^{-1}_{1,\ldots,1}\times id)\,c\mathcal{K}^N_{1,\ldots,1,0}\,(id\times \psi_{1,\ldots,1})$ 
 was firstly considered in \cite{Kassotakis:2021}. Furthermore, since $c\mathcal{K}^N_{1,\ldots,1,0}$ with  $i\in\mathbb{Z}$ (or equivalently $N\rightarrow \infty$) is the non-abelian KP map (Hirota-Miwa map), interesting and open questions concern the underlying geometry and the identification of $c\mathcal{K}^N_{a^1,\ldots,a^N,c}$ as an integrable difference system  when $i\in \mathbb{Z}$ and for generic $a^1,\ldots,a^N$ and $c$.

\appendix
\appendixpage
\section{The non-Abelian $\mathcal{F},\mathcal{K},\mathcal{\Lambda}$ and $\mathcal{H}$ lists} \label{app1}

For generic $a,b,c,$ the families of maps $\mathcal{F}_{a,b,c},$ $\mathcal{K}_{a,b,c},$ $\mathcal{\Lambda}_{a,b,c}$ and $\mathcal{H}_{a,b,c},$   given in Proposition \ref{prop:t} and in (\ref{ybK}) are considered as the generic maps of the $\mathcal{F},\mathcal{K},\mathcal{\Lambda}$ and $\mathcal{H}$ lists. For each generic map various degeneracies can occur by demanding for instance  the coalescence of singularities of the generic map when restricted in the abelian domain. Using as a guiding principle this coalescence of singularities, in what follows we present  the non-Abelian $\mathcal{F}$  list. Then from the $\mathcal{F}-$list and by using Proposition \ref{prop:sec2:1} we obtain the $\mathcal{K},$ $\mathcal{\Lambda}$ and $\mathcal{H}$ lists.

\subsection*{The non-Abelian $\mathcal{F}-$list}
The non-Abelian $\mathcal{F}-$list of quadrirational Yang-Baxter maps reads:
$$
R:(u,p,v,q)\mapsto (x,p,y,q)
$$
where:
\begin{flalign*}
&\begin{aligned}
x=&p\left(p(u-v)(1-v)^{-1}-(q u-p v)(q-v)^{-1} \right)^{-1}\\
  &\quad \qquad \left((u-v)(1-v)^{-1}-(q u-p v)(q-v)^{-1} \right),&\\
y=&q \left( (1-q)u+q-p+(p-1)v\right)\left((1-p)u+(p-q)uv+p(q-1)v \right)^{-1}u,&
\end{aligned}
 &&(\mathcal{F}_I\equiv \mathcal{F}_{1,1,1})\\[7pt]
&\begin{aligned}
x=&q^{-1}(1-v)(u-v)^{-1}(q u-p v+p-q)v(1-v)^{-1},&\\
y=&p^{-1}(q u-p v+p-q)(u-v)^{-1}u,&
\end{aligned}
 &&(\mathcal{F}_{II}\equiv \mathcal{F}_{0,1,1})\\[7pt]
& \begin{aligned}
x=&q^{-1}v(u-v)^{-1}(q u-p v),&\\
y=&p^{-1}(q u-p v)(u-v)^{-1}u,&
\end{aligned}
 &&(\mathcal{F}_{III}\equiv \mathcal{F}_{0,0,1})\\[7pt]
& \begin{aligned}
x=&(u-v)^{-1}(u-v+p-q)v,&\\
y=&(u-v+p-q)(u-v)^{-1}u,&
\end{aligned}
 &&(\mathcal{F}_{IV})\\[7pt]
& \begin{aligned}
x=&v+(p-q)(u-v)^{-1},&\\
y=&u+(p-q)(u-v)^{-1},&
\end{aligned}
 &&(\mathcal{F}_{V})
\end{flalign*}
Note that the $\mathcal{F}_{IV}$ map is obtained from $\mathcal{F}_{II}$ by setting
$$
(x,y,u,v,p,q)\mapsto (1+\epsilon x,1+\epsilon y,1+\epsilon u,1+\epsilon v,1+\epsilon p,1+\epsilon q)
$$
and then sending $\epsilon\rightarrow 0.$ Furthermore the $\mathcal{F}_{V}$ map is obtained from $\mathcal{F}_{IV}$ by setting
$$
(x,y,u,v,p,q)\mapsto (1+\epsilon x,1+\epsilon y,1+\epsilon u,1+\epsilon v,1+\epsilon^2 p,1+\epsilon^2 q)
$$
and then sending $\epsilon\rightarrow 0.$

The non-abelian maps $\mathcal{F}_{III}$ and $\mathcal{F}_{V}$ were first introduced in \cite{ABS:YB}.
\subsection*{The non-Abelian $\mathcal{K}-$list}
The non-Abelian $\mathcal{K}-$list of quadrirational Yang-Baxter maps reads:
$$
R:(u,p,v,q)\mapsto (x,p,y,q)
$$
where:

\begin{flalign*}
&\begin{aligned}
x=&v\left(uv+q(1-u-v)\right)^{-1}\left(p+uv-qu-pv\right),&\\
y=&\left(q+uv-qu+pv\right)\left(uv+p(1-u-v)\right)^{-1}u,&
\end{aligned}
 &&(\mathcal{K}_I\equiv \mathcal{K}_{1,1,1})\\[7pt]
&\begin{aligned}
x=&q^{-1}v\left(1-u-v\right)^{-1}\left(p-q u-pv\right),&\\
y=&p^{-1}\left(q-q u-pv\right)\left(1-u-v\right)^{-1}u,&
\end{aligned}
 &&(\mathcal{K}_{II}\equiv \mathcal{K}_{0,1,1})\\[7pt]
&\begin{aligned}
x=&pv(q u+p v)^{-1}(u+v),&\\
y=&q(u+v)(qu+pv)^{-1}u,&
\end{aligned}
 &&(\mathcal{K}_{III})\\[5pt]
\end{flalign*}

The non-abelian map $\mathcal{K}_{III}$ was first introduced in \cite{Kassotakis:2021}.

\subsection*{The non-Abelian $\mathcal{\Lambda}-$list}
The non-Abelian $\mathcal{\Lambda}-$list of quadrirational Yang-Baxter maps reads:
$$
R:(u,p,v,q)\mapsto (x,p,y,q)
$$
where:

\begin{flalign*}
&\begin{aligned}
x=&pv\left(q u+p v-q (p+u v)\right)^{-1}\left(u+v-q-uv\right),&\\
y=&q\left(u+v-p-uv\right)\left(q u+p v-p(q+uv)\right)^{-1}u,&
\end{aligned}
 &&(\mathcal{\Lambda}_I\equiv \mathcal{\Lambda}_{1,1,1})\\[7pt]
&\begin{aligned}
x=&pv\left(qu+pv-pq\right)^{-1}\left(u+v-q\right),&\\
y=&q\left(u+v-p\right)\left(qu+pv-pq\right)^{-1}u,&
\end{aligned}
 &&(\mathcal{\Lambda}_{II})\\[7pt]
&\begin{aligned}
x=&v(uv+q)^{-1}(uv+p),&\\
y=&(uv+q)(uv+p)^{-1}u,&
\end{aligned}
 &&(\mathcal{\Lambda}_{III})
\end{flalign*}
Note that $\mathcal{\Lambda}_{II}$ is obtained from $(\mathcal{\Lambda}_{a,b,c})$ by setting $c\mapsto a b$ and then taking $a=0.$ We can take $b=0$ instead, but the map we obtain is equivalent up to conjugation with $\chi(\alpha): z\mapsto \alpha z^{-1},$ with  $\mathcal{\Lambda}_{II}$ i.e. this map reads $(\chi(p)\times \chi(q))\, \mathcal{\Lambda}_{II}\,(\chi(p)^{-1}\times \chi(q)^{-1}).$

The non-abelian map $\mathcal{\Lambda}_{III}$ was first introduced in \cite{Doliwa_2014}.

\subsection*{The non-Abelian $\mathcal{H}-$list}
The non-Abelian $\mathcal{H}-$list of quadrirational Yang-Baxter maps reads:
$$
R:(u,p,v,q)\mapsto (x,p,y,q)
$$
where:

\begin{flalign*}
&\begin{aligned}
x=&v\left((uv-q)(v-q)^{-1}-(uv-p)(v-1)^{-1}\right)^{-1}\\
&\quad \qquad\left((uv-q)(v-q)^{-1}-(uv-p)(v-1)^{-1}\right),&\\
y=&\left((1-q)uv+(q-p)v+q(p-1)\right)\left((1-p)uv+(p-q)u+p(q-1)\right)^{-1}u,&
\end{aligned}
 &&(\mathcal{H}_I\equiv \mathcal{H}_{1,1,1})\\[7pt]
&\begin{aligned}
x=&v(uv-q)^{-1}(uv-p),&\\
y=&(uv-q)(uv-p)^{-1}u,&
\end{aligned}
 &&(\mathcal{H}_{III})\\[7pt]
  &\begin{aligned}
x=&v-(p-q)(u+v)^{-1},&\\
y=&u+(p-q)(u+v)^{-1},&
\end{aligned}
 &&(\mathcal{H}_{V})
\end{flalign*}

$\mathcal{H}_V$ has the symmetry $u\mapsto -u.$ Via this symmetry   by using Proposition \ref{prop:sec2:1} we recover
the $\mathcal{F}_V$ map.


\end{document}